\documentclass[11pt,a4paper]{article}
\usepackage{amsfonts,amssymb,amsmath,amsthm,cite}
\usepackage{graphicx}
\usepackage{slashed}

\evensidemargin=\oddsidemargin
\newcommand{\MM}{\mathcal{M}}
\newcommand{\ii}{\mathrm{i}}
\newcommand{\np}{{\mathbf{n}^+}}
\newcommand{\nm}{{\mathbf{n}^-}}
\newcommand{\nn}{\mathbf{n}}
\newcommand{\DD}{\mathcal{D}}
\newcommand{\nchi}{\eta} 

\newtheorem{assumption}{Assumption}[section]
\newtheorem{theorem}[assumption]{Theorem}

\newtheorem{lemma}[assumption]{Lemma}

\newtheorem{remark}[assumption]{Remark}

\begin{document}

\begin{center}

{\Large\textbf{Atiyah-Patodi-Singer Index Theorem for Domain Walls}}
\vspace{0.5cm}

{\large A.~V.~Ivanov$^\dag$ and D.~V.~Vassilevich$^\ddag$}

\vspace{0.5cm}

$^\dag${\it St. Petersburg Department of Steklov Mathematical Institute of
Russian Academy of Sciences,}\\{\it 27 Fontanka, St. Petersburg 191023, Russia}\\
{\it E-mail: regul1@mail.ru}\\
$^\ddag${\it Center of Mathematics, Computation and Cognition, Universidade Federal do ABC}\\
{\it 09210-580, Santo Andr\'e, SP, Brazil}\\
$^\ddag${\it Physics Department, Tomsk State University, Tomsk, Russia}\\
{\it Email: dvassil@gmail.com}

\end{center}

\begin{abstract}
We consider the index of a Dirac operator on a compact even dimensional manifold with a domain wall. The latter is defined as a co-dimension one submanifold where the connection jumps. We formulate and prove an analog of the Atiyah-Patodi-Singer theorem that relates the index to the bulk integral of Pontryagin density and $\eta$-invariants of auxiliary Dirac operators on the domain wall. Thus the index is expressed through the global chiral anomaly in the volume and the parity anomaly on the wall.
\end{abstract}

\section{Introduction}\label{sec:intro}
The Atiyah-Patodi-Singer (APS) index theorem \cite{Atiyah:1975jf} relates the index of a Dirac operator on a manifold with boundary to the integral of the Pontryagin density over the interior of the manifold and to the $\eta$-invariant of an auxiliary Dirac operator on the boundary. This relation is quite remarkable from the point of view of theoretical physics since the Pontryagin density is the local Adler-Bell-Jackiw \cite{Adler:1969gk,Bell:1969ts} axial anomaly while the $\eta$-invariant can be shown to define the parity anomaly \cite{Niemi:1983rq,Redlich:1983dv,AlvarezGaume:1984nf}.

The relations between bulk and boundary anomalies are being intensively studied in the context of the "anomaly inflow" mechanism \cite{Witten:2019bou}. The APS index theorem was used in quantum field theory in \cite{Witten:2015aba} to discuss topological phases in the fermion path integral. 

Although the APS index theorem was incorporated into the quantum field theory context shortly after publication of original APS paper, see \cite{Hortacsu:1980kv}, the application of this theorem were limited by the non-local nature of APS boundary conditions. In more recent work the use of non-local boundary conditions became unnecessary. The paper \cite{Witten:2019bou} analysed the APS theorem with local boundary conditions. The boundary contributions to parity anomaly have been computed in \cite{Kurkov:2017cdz,Kurkov:2018pjw} also for local (bag) boundary conditions.

Another recent and rather fruitful idea was to extend the APS theorem to domain wall type configurations. The paper \cite{Fukaya:2019qlf} defined domain walls as the surfaces where the potential term in Dirac operator is discontinuous. Here we follow the approach of \cite{Vassilevich:2018aqu} where the domain walls are defined as the submanifolds where the connection jumps. This is exactly what happens on domain walls in a ferromagnet. To prove the APS theorem, the paper \cite{Vassilevich:2018aqu} used explicit heat kernel computations of all anomalies in 4D with some restrictions on the geometry near domain walls. The present work extends the results of \cite{Vassilevich:2018aqu} in a very essential way. We lift all restrictions on the dimension and many restrictions on the geometry. Moreover, we use more sophisticated methods to prove the theorem without explicitly computing all terms involved in the equation.

Most of our notations and conventions are borrowed from \cite{Fursaev:2011zz}. Our paper uses a great variety of methods and approaches that are explained in the monographs \cite{BGV2004,BB2013,Gilkey:1984,Kirsten,Nakahara:2003nw}.

This paper is organized as follows. In the next section we introduce main notations and definitions and formulate the principal theorem. This theorem is demonstrated in section \ref{sec:prod} under an additional assumption of product structure near the domain wall. This assumption is lifted in section \ref{sec:lift}. Possible extensions of our main theorem are discussed in section \ref{sec:disc}.

\section{Definitions and statement of the result}\label{sec:def}
Let $\MM$ be a compact orientable Riemannian manifold without boundary of $\dim \MM=n=2m$ and $\Sigma$ be a smooth co-dimension one closed submanifold in $\MM$. Local coordinates on $\MM$ will be denoted by $x^\mu, x^\nu$, etc. Let $x^a$, $a=1,\dots,n-1$, denote coordinates on $\Sigma$. Let us label in an arbitrary way two sides of $\Sigma$ by $+$ and $-$ and denote by $\np$ and $\nm$ two unit normal vectors, see Fig.1. We assume that the metric is smooth across the boundary, so that $\np=-\nm\equiv\nn$. This coordinate system can be extended from $\Sigma$ into $\MM$ as Gaussian coordinates. The coordinate $x^n$ then denotes the geodesic distance to $\Sigma$.  Let $V$ be a hermitian vector bundle over $\MM$. 

The spectral problem which will be considered in this paper is defined by a Dirac operator\footnote{Note that the existence of a Dirac operator does not necessarily require any spin structure. An example is the Hodge-Dirac operator $d+\delta$ defined on differential forms.} in the bulk jointly with matching conditions at the interface surface $\Sigma$.
On $\MM \backslash \Sigma$ the Dirac operator is given by the local formula
\begin{equation}
\slashed{D}=\ii \gamma^\mu \nabla_\mu \,,\label{Dirop}
\end{equation}
where the $\gamma$-matrices act fibre-wise on $V$, are hermitian, and satisfy the Clifford relation
$\gamma^\mu\gamma^\nu +\gamma^\nu\gamma^\mu =2g^{\mu\nu}I_V$.
The connection has to be compatible with the Clifford structure, $\nabla_\mu \gamma^\nu = \gamma^\nu \nabla_\mu$. Let us introduce a chirality matrix
\begin{equation}
\gamma_* = -\frac{\ii^m}{n!} \epsilon_{\mu\dots\rho}\gamma^\mu \dots \gamma^\rho \label{gstar}
\end{equation}
with $\epsilon_{\mu\dots\rho}$ being the Levi-Civita tensor on $\MM$. Obviously, $\gamma_*\gamma^\mu = -\gamma^\mu \gamma_*$, $\gamma_*^2=I_V$, $\gamma_*$ is hermitian, $\gamma_*^\dag=\gamma_*$, and
\begin{equation}
\slashed{D}\gamma_*=-\gamma_*\slashed{D} .\label{Dg}
\end{equation}
The bundle $V$ can be split into $V=V_R\oplus V_L$, so that $\gamma_*$ equals to $1$ (respectively, to $-1$) on $V_R$ (respectively, on $V_L$). According to the physics tradition, the sections $\psi\in \Gamma(V)$ will be called spinors, while the sections of $V_R$ and $V_L$ will be called right and left (chiral) spinors, respectively.

We assume that covariant derivatives $ \nabla_a$ have well defined though possibly different limits on $\Sigma$, $ \nabla_a^+$ and $ \nabla_a^-$, respectively, from different sides. However, since they both define connections on the same vector bundle (that is a restriction of $V$ to $\Sigma$) their difference,
\begin{equation}
B_a:= \nabla_a^+ - \nabla_a^- ,\label{Ba}
\end{equation}
is a vector.

A natural matching condition on $\Sigma$ is that the spinors are continuous,
\begin{equation}
\psi^+=\psi^- \label{mc1}
\end{equation}
(we continue to denote by superscripts $+/-$ the directional limits of various quantities on $\Sigma$ taken from different sides). Since we are going to use a second order operator $\slashed{D}^2$, we need a second matching condition that follows from (\ref{mc1}), that is $(\slashed{D}\psi)^+=(\slashed{D}\psi)^-$. More explicitly,
\begin{equation}
(\nabla_\np \psi)^+ + (\nabla_\nm \psi)^-=-\gamma^\nn \gamma^aB_a\psi\vert_{\Sigma} .\label{mc2}
\end{equation}

Note, that the conditions (\ref{mc1}) and (\ref{mc2}) respect chirality.  Namely, $\psi$ satisfies (\ref{mc1}) and (\ref{mc2}) iff $\gamma_*\psi$ satisfies these two matching conditions. Thus, by taking into account (\ref{Dg}) we conclude that $\slashed{D}$ splits into $\slashed{D}_R:\Gamma(V_R)\to \Gamma(V_L)$ and $\slashed{D}_L:\Gamma(V_L)\to \Gamma(V_R)$ and $\slashed{D}_L=\slashed{D}_R^\dag$. I.e., we have an elliptic complex and thus can define the Index of $\slashed{D}$ as the difference between numbers of zero eigenmodes of $\slashed{D}$ with right and left chiralities. As usual, this index may also be represented through an $L^2$ trace involving the heat operator of $\slashed{D}^2$,
\begin{equation}
\mbox{Index}\, (\slashed{D}) = \mathrm{Tr} \left( \gamma_* e^{-t\slashed{D}^2} \right), \qquad t>0 . \label{IndD}
\end{equation}

Let $e_\alpha$, $\alpha=1,\dots,n$, be a local orthonormal basis in the tangent bundle $T\MM$ and let $e^\alpha$ be a dual basis in $T^*\MM$. The corresponding Levi-Civita connection reads component-wise $w^{\alpha\beta}_{\mu}=e^\nu_\beta \bigl( \Gamma_{\mu\nu}^\rho e_\rho^\alpha -\partial_\mu e_\nu^\alpha \bigr)$. Here $\Gamma_{\mu\nu}^\rho$ is the Christoffel symbol. The flat indices $\alpha$, $\beta$, etc are lowered and raised with the unit matrix. Thus, their particular position (up or down) does not play any role.  The corresponding spin-connection is then defined locally as
\begin{equation}
w_\mu^{[s]}:=\tfrac 18 w_\mu^{\alpha\beta}[\gamma_\alpha,\gamma_\beta],\qquad
\nabla_\mu^{[s]}=\partial_\mu +\omega_\mu^{[s]}.\label{spincon}
\end{equation}
The field
\begin{equation}
A:=\nabla - \nabla^{[s]} \label{defA}
\end{equation}
will be called the Yang-Mills connection. $A$ is a diffeomorphism vector. It satisfies $[A,\gamma^\mu]=0$. 

Near $\Sigma$ our choice of the frame is consistent with the Gaussian coordinate system: $e_n^\nn=1$, $e_A^\nn=e_n^a=0$ for $A=1,\dots,n-1$. Since $\Gamma_{a\nn}^b=-K_a^b$ on $\Sigma$, with $K_a^b$ being the extrinsic curvature, $w_a^{An}=-K^A_a=-K_a^b e_b^A$. Using the Gauss coordinate system we can extend these relations to some neighbourhood of $\Sigma$.

Since all irreducible representations of the Clifford algebra in an even-dimensional space of Euclidean signature are unitary equivalent, by a suitable choice of the basis one can transform the $\gamma$-matrices to the form: 
\begin{equation}
\gamma^a=\left( \begin{array}{cc} 0 & 1 \\ 1 & 0 \end{array} \right) \otimes \hat\gamma^a,\qquad
\gamma^\nn=\left( \begin{array}{cc} 0 & \ii \\ -\ii & 0 \end{array} \right) \otimes \mathrm{id},\qquad
\gamma_*=\left( \begin{array}{cc} 1 & 0 \\ 0 & -1 \end{array} \right) \otimes \mathrm{id}.\label{newgammas}
\end{equation}
The matrices $\hat \gamma^a$ have to satisfy
\begin{equation*}
\epsilon_{a\dots c\nn}\hat\gamma^a \dots \hat\gamma^c =-(n-1)! (-\ii)^{m-1}.
\end{equation*}
This condition fixes one of the two inequivalent representations of the Clifford algebra in odd dimensions.

The Yang-Mills field and $B_a$ have to be diagonal in this basis
\begin{equation*}
A_\mu =\left( \begin{array}{cc} \hat A_\mu & 0 \\ 0 & \hat A_\mu 
\end{array} \right), \qquad B_a=\left( \begin{array}{cc} \hat B_a & 0 \\ 0 & \hat B_a 
\end{array} \right)
\end{equation*}
with $[\hat A_\mu, \hat \gamma_a]=0$. By using the symmetry of extrinsic curvature one brings the Dirac operator to the following form
\begin{equation}
\slashed{D}=\left( \begin{array}{cc} 0 & -\hat\nabla_\nn + \DD +\tfrac 12 K_a^a \\
\hat\nabla_\nn +\DD-\tfrac 12 K_a^a & 0 \end{array} \right), \qquad \DD=i\hat\gamma^a \hat\nabla_a \,,\label{DDgen}
\end{equation} 
where $\hat\nabla$ is defined by the local formula
\begin{equation}
\hat\nabla_\mu=\partial_\mu +\hat A_\mu + \tfrac 18 w_\mu^{AB}[\hat\gamma_A,\hat\gamma_B].
\end{equation}
We see, that the operator $\DD$ is a Dirac operator on $\Sigma$ depending on the spin-connection defined by the frame $e^A$ and the Yang-Mills field $\hat A_a$. However, the Yang-Mills field has different limits on $\Sigma$, $A_a^+$ and $A_a^-$, such that
\begin{equation*}
\hat A_a^+ - \hat A_a^-=\hat B_a \,.
\end{equation*}
In accordance to this, the operator $\DD$ has different limits $\DD^\pm$ as well, and
\begin{equation}
\DD^+ - \DD^-= \ii\hat \gamma^a \hat B_a.\label{DpDm}
\end{equation}

It will be convenient to impose the gauge condition
\begin{equation}
\omega_n=0 \label{axial}
\end{equation}
on all components (spin and Yang-Mills) of the connection. This gauge condition is admissible locally in a vicinity of $\Sigma$.

In what follows we shall need spectral functions of the operators $\slashed{D}$, $\DD^+$, and $\DD^-$.

Let $L$ be a Laplace type operator on a $d$-dimensional manifold, and let $Q$ be a smooth endomorphism (locally - a smooth matrix valued function). Then there is a full asymptotic expansion at $\tau \to +0$
\begin{equation}
\mathrm{Tr} \left( Q e^{-\tau  L} \right) \simeq \sum_{k=0}^\infty \tau^{\frac{k-d}2} a_k(Q,L). \label{asymex}
\end{equation}
For domain walls with matching conditions of the type considered here this property was established in \cite{Bordag:1999ed,Moss:2000gv,Gilkey:2001mj}. The heat kernel coefficients appear to be local, i.e. they are given by integrals over $\MM$ and $\Sigma$ of polynomials constructed from local invariants associated with the problem. For the bulk integral over $\MM$ these are the Riemannian and Yang-Mills curvatures and their covariant derivatives. On $\Sigma$, one also allows the extrinsic curvature and the vector $B$.

Since the left hand side of (\ref{IndD}) does not depend on $t$,
\begin{equation}
\mbox{Index}\, (\slashed{D})=a_n(\gamma_\star,\slashed{D}^2).\label{IndD2}
\end{equation}
Let $f$ be a smooth function with a compact support that does not intersect $\Sigma$. The Pontryagin density $P(x)$ for $\slashed{D}$ is defined through the equation
\begin{equation}
\int_{\MM} d^n x\, f(x)P(x)=a_n(f\gamma_*,\slashed{D}^2). \label{Pont}
\end{equation}

The $\eta$ function of $\DD$ is defined as a sum over the eigenvalues $\lambda$,
\begin{equation}
\eta (z ,\DD)=\sum_{\lambda>0} \lambda^{-z } - \sum_{\lambda<0} (-\lambda)^{-z } .\label{etas}
\end{equation}

Let us define a smooth family of operators $\DD(s)$ such that $\DD(0)=\DD^-$ and $\DD(l)=\DD^+$. We assume that $\DD(s)=\ii \hat \gamma^a\hat \nabla_a(s)$ and for all values of $s\in [0,l]$ the connection remains compatible with the Clifford structure on $\Sigma$.   For the future use, let us also assume that at the endpoints of the interval $[0,l]$ all derivatives of $\DD(s)$ with respect to $\tau$ vanish.

The value of $\eta(z,\DD)$ at $z=0$ measures the spectral asymmetry of $\DD$. $\eta(0,\DD(s))$ is discontinuous when an eigenvalue passes through $0$. If no eigenvalue of $\DD(s)$ changes its sign when $s$ varies between $0$ and $l$, we have
\begin{equation}
\partial_s \eta(0,\DD(s))= -\frac 2{\sqrt{\pi}} a_{n-2}(\partial_s \DD(s),\DD(s)^2), \label{dertau}
\end{equation}
see \cite{Gilkey:1984}.
We use the right hand side of this formula to define the \emph{relative spectral asymmetry} as
\begin{equation}
\tilde\eta(\DD^+,\DD^-)= -\frac 2{\sqrt{\pi}} \int_0^l ds\, a_{n-2}(\partial_s\DD(s),\DD(s)^2).\label{tildeta}
\end{equation}
This definition will be used independently on whether the eigenvalues cross $0$ or not. We stress, that in general the right hand side of (\ref{tildeta}) is \emph{not} equal to $\eta(0,\DD^+) - \eta(0,\DD^-)$. In contrast to $\eta(0,\DD)$, the relative spectral asymmetry is always local.

We are ready to formulate the main result of this work

\begin{theorem} \label{mainT}
Suppose that at least one of the following assumptions holds true:\\ 
(a) the extrinsic curvature $K_a^b$ of $\Sigma$ vanishes;\\ 
(b) $n\leq 4$; \\ 
(c) $n=6$ and $\mathrm{tr}(F^+_{ab}-F^-_{ab})=0$.\\ \noindent
Then
\begin{equation}
\mathrm{Index}(\slashed{D})=\int_{\MM\backslash\Sigma} d^nx\, P(x) - \tfrac 12 \tilde\eta(\DD^+,\DD^-). \label{APST}
\end{equation}
\end{theorem}

\begin{remark}
For generic situations when none of the assumptions (a) - (c) holds, the right hand side of (\ref{APST}) has to contain a correction term given in (\ref{TA}).
\end{remark}

We shall prove Theorem \ref{mainT} in a two-step procedure. In the next section we deal with the case of a product structure near $\Sigma$. This assumption is lifted in section \ref{sec:lift}.
\section{The case of product structure near $\Sigma$}\label{sec:prod}
Here we consider the case when $\MM$ has a product structure $\Sigma \times [-\varepsilon_0,\varepsilon_0]$ near the interface surface as depicted at Fig.\ \ref{fig1}, $s:=x^n$. We assume that the operator $\slashed{D}$ also has this structure, so that on each side of $\Sigma$
\begin{equation}
\slashed{D}=\left( \begin{array}{cc} 0 & -\partial_s + \DD \\
\partial_s +\DD & 0 \end{array} \right), \qquad \DD=i\hat\gamma^a \hat\nabla_a \label{DD}
\end{equation} 
with $\DD=\DD^+$ for $s\in (0,\varepsilon_0]$ and $\DD =\DD^-$ for $s\in [-\varepsilon_0,0)$.
\begin{figure}[h]
\centerline{\includegraphics[width=0.8\linewidth]{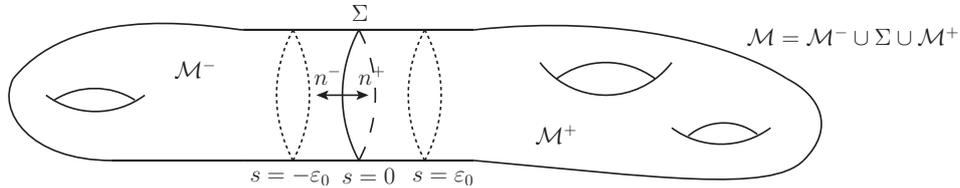}}
\caption{Manifold $\MM$ having a product structure near $\Sigma$.}
\label{fig1}
\end{figure}

To prove Theorem \ref{mainT} for this case we use the following method. We cut $\MM$ along $\Sigma$ and paste a cylinder $\mathcal{C}=\Sigma \times [0,l]$ as depicted on Fig.\ \ref{fig2}. The manifold obtained in this way will be denoted as $\widetilde{\MM}$. On the cylinder, we define 
\begin{equation}
\DD(s)\equiv \DD^- + f(s/l) \ii \hat \gamma^a \hat B_a, \qquad \mbox{for}\ s\in [0,l],
\label{extDD}
\end{equation}
where $f$ is a smooth smearing function s.t. $f(0)=0$, $f(1)=1$, and $f^{(m)}(0)=f^{(m)}(1)=0,\ \forall m>0$.

\begin{figure}[h]
\centerline{\includegraphics[width=0.7\linewidth]{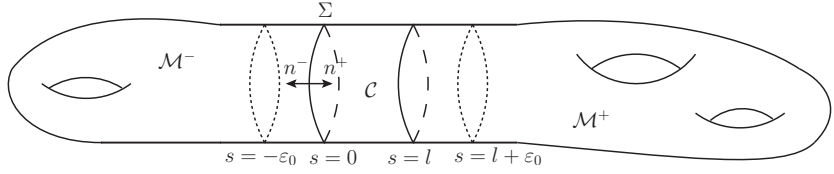}}
\caption{The extended manifold $\widetilde{\MM}$.}
\label{fig2}
\end{figure}

For the future use, we derive
\begin{equation}
\label{w2}
\slashed{D}^2=
\begin{pmatrix}
-\partial_s^2+\DD^2(s)-\dot{\DD}(s)&0\\
0&-\partial_s^2+\DD^2(s) +\dot{\DD}(s)
\end{pmatrix}
,
\end{equation}
where the dot denotes derivative with respect to $s$.

Thus, we obtained a smooth eigenvalue problem on $\widetilde{\MM}$. The corresponding Index reads
\begin{equation}
\mbox{Index}\,  (\slashed{D})_{\widetilde{\MM}}=\int_{\widetilde{\MM}} d^nx\, P(x)=
\int_{{\MM}} d^nx\, P(x) +\int_{\mathcal{C}} d^nx\, P(x). \label{Intild}
\end{equation}
In the last integral above $P(x)$ is the Pontryagin density for the new spectral problem on $\widetilde{\MM}$ rather then for the initial problem. We do not introduce a new notation as this is not likely to create a confusion.

First, we demonstrate 
\begin{lemma}\label{l31}
Under the conditions established above
\begin{equation}
\int_{\mathcal{C}}d^nx\,P(x)=-\tfrac{1}{2}\tilde{\eta}(\DD^+,\DD^-).
\end{equation}
\end{lemma}

\begin{proof}
Let us consider a fundamental solution $K(x,y;\tau)$ of the heat equation (the heat kernel) which satisfies 
\begin{equation}
\label{w1}
\left( \frac{\partial}{\partial\tau}+\slashed{D}^2\right)K({x},{y};\tau)=0,\qquad
K(x,y;0)=g^{-\frac{1}{2}}({x})\delta({x}-{y}).
\end{equation}
To construct an asymptotic expansion of the heat kernel we shall use the iterative procedure proposed by DeWitt \cite{DeWitt1965} (with firmer mathematical grounds to be found in \cite{BGV2004} and further development of the method -- in \cite{BV}). If both arguments $x$ and $y$ belong to $\mathcal{C}$, the expansion reads
\begin{equation}
\label{t1}
K({x},{y};\tau)=\tau^{-\frac{n}{2}}\Delta^{1/2}(\hat x,\hat y)
e^{-\frac{\sigma(\hat x,\hat y)}{2\tau}-\frac{(s-t)^2}{4\tau}}\sum\limits_{k=0}^{\infty}\tau^ka_{2k}(\slashed{D}^2;{x},{y}),
\end{equation}
where we used the notations $x=(\hat x,s)$, $y=(\hat y,t)$ with $s,t\in [0,l]$ and $\hat x, \hat y\in \Sigma$. $\sigma(\hat x,\hat y)$ is the Synge's world function on $\Sigma$, i.e. half the square of geodesic distance between $\hat x$ and $\hat y$, and $\sigma^a(\hat x,\hat y)\equiv \partial_{\hat x^a} \sigma (\hat x,\hat y)$.  $\Delta^{1/2}$ is the Van-Vleck-Morette determinant 
\begin{equation*}
\Delta^{1/2}(\hat x,\hat y)=[g(\hat x)g(\hat y)]^{-1/2}\det\left[-\frac{\partial^2\sigma(\hat x,\hat y)}{\partial x^a\partial y^b}\right].
\end{equation*}
Since $\widetilde\MM$ has no boundaries or domain walls, only even-numbered heat kernel coefficients appear in (\ref{t1}). They satisfy the following iterative equation
\begin{multline}
\label{1}
\begin{pmatrix}
k+\sigma^a\hat\nabla_a(s)+(s-t)\partial_s&0\\
0&k+\sigma^a\hat\nabla_a(s)+(s-t)\partial_s
\end{pmatrix}
a_{2k}(\slashed{D}^2;{x},{y})=\\=\Delta^{-1/2}(\hat x,\hat y)
\begin{pmatrix}
\partial_s^2-\DD(s)^2+\dot{\DD}(s)&0\\
0&\partial_s^2-\DD(s)^2-\dot{\DD}(s)
\end{pmatrix}
\Delta^{1/2}(\hat x,\hat y)
a_{2k-2}(\slashed{D}^2;{x},{y})
\end{multline}
for $k>0$. The coefficient $a_{0}(\slashed{D}^2;{x},{y})$ satisfies the same relation with $0$ on the right hand side and the initial condition 
\begin{equation}
a_{0}(\slashed{D}^2;{x},{x})=(4\pi)^{-\frac n2} \mathrm{id}.\label{a0xx}
\end{equation}

Due to Eq.\ (\ref{Intild}), the intergal of $P(x)$ over $\mathcal{C}$ cannot depend on the length $l$ of the cylinder. Thus, in the $l\to \infty$ limit only the terms in $P(x)$ that behave as $1/l$ contribute to the index. Since 
\begin{equation}
P(x) = g^{1/2} \mathrm{tr}\bigl( \gamma_* a_n(\slashed{D}^2;x,x)\bigr) \label{Pan}
\end{equation}
also in the heat kernel coefficients we have to keep the $1/l$ terms only. The simplest way to analyse the behaviour of various terms in Eq.\ (\ref{1}) in this limit is to rescale the coordinate $s\to s/l$ which has to be accompanied by the change of the metric $g_{nn}=1\to l^2$ and $g^{nn}=1\to l^{-2}$. The net result is that that $\partial_s^2$ receives a factor of $l^{-2}$ while $\dot\DD$ changes $l^{-1}\dot \DD$. Other terms remain bounded. One of the consequences of this scaling behavior is that the term $\partial_s^2$  can be omitted on the right hand side of Eq.\ (\ref{1}) as long as one is interested in the terms that survive in the limit $l\to \infty$ after the integration over $\mathcal{C}$. Thus, one can go to the coincidence limit $t=s$ right in the equation (\ref{1}) to obtain
\begin{multline}
\label{lminus1}
\begin{pmatrix}
k+\sigma^a\hat\nabla_a(s)&0\\
0&k+\sigma^a\hat\nabla_a(s)
\end{pmatrix}
a_{2k}(\slashed{D}^2;({\hat x},s),(\hat{y},s))=\\=\Delta^{-1/2}(\hat x,\hat y)
\begin{pmatrix}
-\DD(s)^2+\dot{\DD}(s)&0\\
0&-\DD(s)^2-\dot{\DD}(s)
\end{pmatrix}
\Delta^{1/2}(\hat x,\hat y)
a_{2k-2}(\slashed{D}^2;({\hat x},s),(\hat{y},s)).
\end{multline}
Here one recognizes the recurrence relation for the heat kernel coefficients of the operator $\mathcal{L}(s,1)$, where
\begin{equation}
\mathcal{L}(s,\xi)=\begin{pmatrix}
\DD(s)^2-\xi \dot{\DD}(s)&0\\
0&\DD(s)^2+\xi \dot{\DD}(s)
\end{pmatrix}. \label{f30}
\end{equation}
Thus, modulo the terms that vanish faster than $l^{-1}$ in the limit $l\to \infty$,
\begin{equation}
a_{n}(\slashed{D}^2;({\hat x},s),(\hat{y},s))=(4\pi)^{-\frac 12}
a_{n}(\mathcal{L}(s,1);\hat x,\hat y).\label{a2ka2k}
\end{equation}
The operator $\mathcal{L}$ is defined over $\Sigma$. Thus, some formulas for the heat kernel expansion require (quite obvious) modification. In particular, one has to replace $n$ by $n-1$ in (\ref{t1}) and  (\ref{f30}).
Moreover, since $\dot\DD(s)\sim l^{-1}$ we have to keep only the terms that are linear in $\dot\DD(s)$. Therefore,
\begin{equation*}
\int_{\mathcal{C}}d^nx\, g^{1/2}\, \mathrm{tr}\, \bigl( \gamma_* a_n(\slashed{D}^2; x,x)\bigr)=
(4\pi)^{-\frac 12} \left. \frac{d}{d\xi}\right\vert_{\xi=0} \int_0^lds\int_\Sigma d^{n-1}\hat x\, g^{1/2}\, \mathrm{tr}\, \bigl( \gamma_* a_n(\mathcal{L}(s,\xi);\hat x,\hat x) \bigr).
\end{equation*}
The derivative with respect to $\xi$ can be computed by expanding in $\tau$ the following identity
\begin{eqnarray}
&& \left. \frac{d}{d\xi}\right\vert_{\xi=0} \mathrm{Tr}\, \bigl( \gamma_* \exp (-\tau \mathcal{L}(s,\xi))\bigr)=-\tau \mathrm{Tr}\, \left( \gamma_* \frac{d \mathcal{L}(s,\xi)}{d\xi} \exp (-\tau \mathcal{L}(s,0)) \right)\nonumber \\
&&\qquad\qquad  = -\tau \mathrm{Tr}\, \left( \begin{pmatrix} 1 & 0 \\ 0 & -1
\end{pmatrix} \begin{pmatrix} -\dot\DD(s) & 0 \\ 0 & \dot \DD(s)
\end{pmatrix} \exp (-\tau \mathcal{L}(s,0)) \right)\nonumber\\
&&\qquad\qquad  = 2\tau \mathrm{Tr}\, \left( \dot\DD(s)\, \exp (-\tau\DD(s)^2)\right),
\nonumber
\end{eqnarray}
where we used that $\gamma_*$ commutes with $\mathcal{L}(s,0)$. This yields
 \begin{equation*}
\int_{\mathcal{C}}d^nx\, g^{1/2}\, \mathrm{tr}\, \bigl( \gamma_* a_n(\slashed{D}^2; x,x)\bigr)= \frac 1{\sqrt{\pi}}\int_0^l ds\, a_{n-2} (\dot\DD(s),\DD(s)^2)
\end{equation*}
and, after taking into account (\ref{Pan}) and (\ref{tildeta}), completes the Proof.
\end{proof}

An alternative proof of this Lemma can be obtained by using the expansion from Ref.\ \cite{2}.

To complete the proof of Theorem \ref{mainT} it still remains to show that smoothing of the spectral problem by pasting a cylinder does not change the index. First, we consider a different smoothing. Let us take some $\delta_0$ such that $0<\delta_0<\varepsilon_0$ and a family functions $\chi^\delta (s)$ with $\delta\in [0,1]$. We assume that $\chi^\delta(s)$ interpolates between $0=\chi^\delta (s)\vert_{s\leq -\delta_0}$ and $1=\chi^\delta (s)\vert_{s\geq \delta}$, is smooth for $s\neq 0$, and
\begin{equation*}
\chi^\delta(-0)=1-\chi^\delta (+0)=\delta/2,
\end{equation*}
 see Fig. \ref{Fig7}. We define a connection 
\begin{equation}
A_a^\delta(x)=A_a^-(\hat x) +\chi^\delta (s) B_a(\hat{x})\label{Aa}
\end{equation}
and an operator $\slashed{D}_\delta$ whose symbol coincides with that of $\slashed{D}$ outside $\Sigma\times [-\delta_0,\delta_0]$ and has the connection (\ref{Aa}) inside $\Sigma\times [-\delta_0,\delta_0]$. For $\delta=0$ we have our initial spectral problem, while for $\delta=1$ the Dirac operator on $\widetilde{\MM}$ and $\slashed{D}_\delta$ are related through an (obvious) smooth homotopy.

\begin{figure}[h]
\centerline{\includegraphics[width=0.7\linewidth]{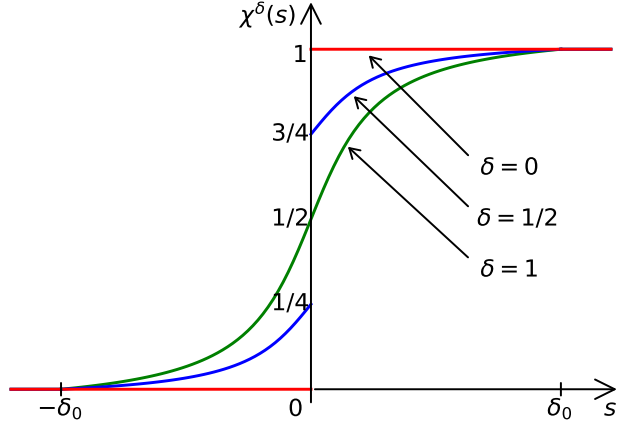}}
\caption{The function $\chi^\delta(s)$.}
\label{Fig7}
\end{figure}

Since the Dirac operator on $\widetilde{\MM}$ and $\slashed{D}_\delta$ are related through a smooth homotopy,
\begin{equation}
\mbox{Index}\,  (\slashed{D})_{\widetilde{\MM}}=\mbox{Index}\, (\slashed{D}_\delta).
\label{IIdel}
\end{equation}

\begin{lemma}\label{l32}
For the family of operators $\slashed{D}_\delta$, $\delta\in [0,1]$, one has
\begin{equation}
  \mathrm {Index}\, (\slashed{D}_\delta) =\mathrm{ Index}\,  (\slashed{D}).
\end{equation}
\end{lemma}

\begin{proof}
Under the homotopy described in Eq.\ (\ref{Aa}) all curvatures and their derivatives in the bulk, as well as the matching conditions on $\Sigma$ together with relevant geometric invariants on $\Sigma$ are smooth functions of $\delta$. Thus, for any $Q$ and any $k\geq 0$ the coefficient $a_k(Q,\slashed{D}_\delta^2)$ is also a smooth function of $\delta$. By taking $Q=\gamma_*$ and $k=n$, and by remembering that the Index is an integer, we complete the proof of this Lemma.
\end{proof}

By combining Eq.\ (\ref{Intild}), Lemma \ref{l31}, and Eq.\ (\ref{IIdel}) with Lemma \ref{l32} one completes the proof of Theorem \ref{mainT} in the case of a product structure near $\Sigma$.

\section{Lifting the assumption of product structure near $\Sigma$}\label{sec:lift}
Consider some geometry that does not satisfy the product structure assumption from the previous section. Let us deform smoothly the metric and the connection near $\Sigma$ without changing limiting values of the connection on $\Sigma$ from both sides and of the induced metric on $\Sigma$ to transform given geometry to a product geometry in a small but finite neighbourhood of $\Sigma$. This can be done in the following way. Fix an $\varepsilon>0$ such that in the $\varepsilon$-vicinity of $\Sigma$ one can introduce the Gaussian normal coordinates and impose the axial gauge (\ref{axial}). Take a smooth function $\eta^1(s)$ satisfying $\eta^1(s)=0$ for $0\leq s \leq \varepsilon_1$ with some $\varepsilon_1<\varepsilon$ and $\eta^1(s)=s$ for $s\geq \varepsilon$, see Fig.\ \ref{figeta}. Define a smooth family of functions
\begin{equation}
\nchi^\delta (s)=s (1-\delta) + \nchi^1(s)\delta 
\end{equation}
depending on $\delta\in [0,1]$. Let us deform the fields on $\Sigma\times [-\varepsilon,\varepsilon]$ as
\begin{eqnarray}
&e_a^{A\delta}(s,\hat x)=e_a^A (\eta^\delta(s),\hat x),\qquad
&e_a^{A\delta}(-s,\hat x)=e_a^A (-\eta^\delta(s),\hat x),\nonumber \\
&A_a^\delta(s,\hat x)=A_a (\eta^\delta(s),\hat x),\qquad
&A_a^\delta(-s,\hat x)=A_a (-\eta^\delta(s),\hat x), \label{Hometa}
\end{eqnarray}
$s\in [0,\varepsilon ]$.
The Levi-Civita connection is deformed accordingly. At $\delta=1$ we have a product geometry in the cylinder $\Sigma \times [-\varepsilon_1,\varepsilon_1]$. Hence, the APS Theorem \ref{mainT} holds for $\delta=1$.
\begin{figure}[h]
\centerline{\includegraphics[width=0.7\linewidth]{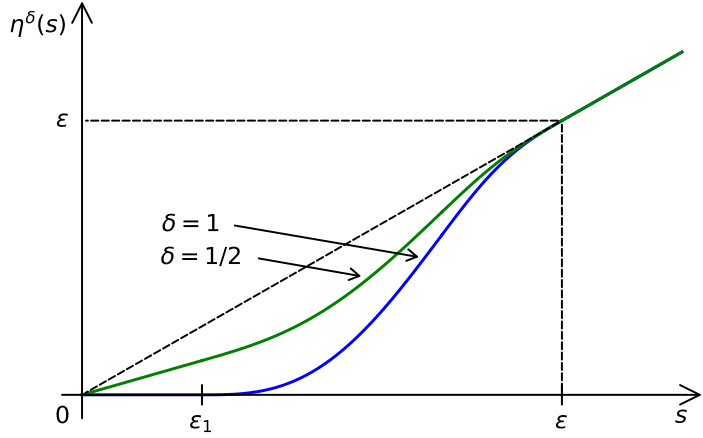}}
\caption{The family of functions $\eta^\delta(s)$.}
\label{figeta}
\end{figure}

The deformation (\ref{Hometa}) induces a smooth homotopy of the symbol of $\slashed{D}$ which preserves the boundary data. Thus, it  does not change the index and does not change the operator $\DD$. The integral of Pontryagin density over $\MM\backslash\Sigma$ can change, however. Thus the assertion of Theorem \ref{mainT} holds if and only if the integral of Pontryagin density also remains unchanged.

To proceed, we need an explicit expression for the Pontryagin density \cite{Nakahara:2003nw}
\begin{equation}
P(x) d^n x=\widehat{\mathrm{A}}(\mathcal{R})\wedge \mathrm{ch}(\mathcal{F}) \label{PAch}
\end{equation}
through the $\widehat{\mathrm{A}}$-genus and the total Chern character computed on the curvature 2-form $\mathcal{R}$ and the Yang-Mills field strength 2-form $\mathcal{F}$. Denote $\mathcal{C}_+=\Sigma \times (0,\varepsilon]$ and $\mathcal{C}_-=\Sigma \times [-\varepsilon,0)$. Then
\begin{eqnarray}
&& \int_{\MM \backslash \Sigma} d^nx\, (P^\delta -P)=\int_{\mathcal{C}_+\cup\, \mathcal{C}_-} \bigl( \widehat{\mathrm{A}}(\mathcal{R}^\delta)\wedge \mathrm{ch}(\mathcal{F}^\delta)-
\widehat{\mathrm{A}}(\mathcal{R})\wedge \mathrm{ch}(\mathcal{F})\bigr) \nonumber \\
&& \qquad =\int_{\mathcal{C}_+\cup\, \mathcal{C}_-} \bigl( [\widehat{\mathrm{A}}(\mathcal{R}^\delta)-\widehat{\mathrm{A}}(\mathcal{R})]\wedge \mathrm{ch}(\mathcal{F}^\delta) -
\widehat{\mathrm{A}}(\mathcal{R})\wedge [\mathrm{ch}(\mathcal{F}^\delta)-\mathrm{ch}(\mathcal{F})] \bigr). \label{2046}
\end{eqnarray}
The $\widehat{\mathrm{A}}$-genus and the Chern class are characteristic polynomials. Thus they are closed and their differences due to a homotopy are exact forms
\begin{equation}
\widehat{\mathrm{A}}(\mathcal{R}^\delta)-
\widehat{\mathrm{A}}(\mathcal{R})=
\mathrm{d}T\widehat{\mathrm{A}}(\Gamma^\delta,\Gamma),\quad
\mathrm{ch}(\mathcal{F}^\delta)-\mathrm{ch}(\mathcal{F})=\mathrm{d}T\mathrm{ch}(A^\delta,A),
\end{equation}
where $T$ means the transgression of corresponding invariant polynomials. Hence, the integrals in (\ref{2046}) are surface terms. Precise definition can be found in \cite{Nakahara:2003nw}. We need just one basic property of the transgression, namely that $T\mathrm{ch}(A^\delta,A)$ vanishes for coinciding connections $A$ and $A^\delta$. Since $A=A^\delta$ on all boundaries of ${\mathcal{C}_+\cup \mathcal{C}_-}$, the second term on the second line of (\ref{2046}) actually vanishes. The whole expression reads
\begin{equation}
\int_\Sigma T\widehat{\mathrm{A}}(\Gamma^\delta,\Gamma)\wedge [\mathrm{ch}(\mathcal{F}^+)-
\mathrm{ch}(\mathcal{F}^-)]\label{TA}
\end{equation}
where we took into account that $\Gamma=\Gamma^\delta$ on $\Sigma \times \{ \varepsilon \}$ and on $\Sigma \times \{ -\varepsilon \}$. 

The assertion of Theorem \ref{mainT} holds if and only if the expression (\ref{TA}) vanishes for $\delta=1$. This can be guaranteed in several cases. It is easy to check, that if $K_{ab}=0$ for $\delta=0$, then $(\Gamma_a - \Gamma_a^1)\mathrm{d}x^a=0$, and thus (\ref{TA}) vanishes. This corresponds to the case (a) of Theorem \ref{mainT}. Since the $\widehat{\mathrm{A}}$-genus contains $4k$-forms, the expression (\ref{TA}) vanish trivially in $n=2$. In $n=4$, only the 0-form $\mathrm{ch}_0$ may contribute to (\ref{TA}). However, since $\mathrm{ch}_0$ does not depend on the connection, (\ref{TA}) vanishes in $n=4$ as well. This is the assumption (b) of Theorem \ref{mainT}. Similarly, in $n=6$ just $\mathrm{ch}_1(\mathcal{F}^+)-\mathrm{ch}_1(\mathcal{F}^-)$ contributes. Due to restriction on $\Sigma$ this last expression becomes proportional to $\mathrm{tr}(F_{ab}^+ - F_{ab}^-)$ and yields the condition (c) of Theorem \ref{mainT}. Note, that this last expression vanishes if the Yang-Mills connection corresponds to a compact semisimple group. 

This completes the proof of our main Theorem \ref{mainT}.
 
\section{Discussion}\label{sec:disc}
As we have already wrote above, the APS index theorem relates chiral and parity anomalies and thus provides an important physical information. In this last section we discuss the prospects of generalizing and extending our main Theorem \ref{mainT}. First of all, the detailed necessary and sufficient conditions is the vanishing of (\ref{TA}). At the moment, we do not see other tractable general cases of vanishing (\ref{TA}) besides the ones that are listed in Theorem \ref{mainT}. However, in dealing with some particular examples in higher dimensions the expression (\ref{TA}) can possibly be instrumental.

Another possibility is to consider a more general operator of Dirac type (containing an axial vector field, e.g.) or to generalize the gluing conditions of $\Sigma$ (for instance, by allowing a brane-world type non-smooth metrics). In both cases one needs an educated guess for the pair of operators to be used instead of $\DD^+$ and $\DD^-$. Explicit computations in $n=4$ along the line of \cite{Vassilevich:2018aqu} are possible, but they do not give a sufficient insight.

\paragraph{Acknowledgements.} The work of A.V.I. was supported by the Russian Science Foundation (project 19-11-00131). Also, A.V.I. is a winner of the Young Russian Mathematician contest and would like to thank its sponsors and jury.
The work of D.V.V. was supported in parts by the S\~ao Paulo Research Foundation (FAPESP), project 2016/03319-6, by the grant 305594/2019-2 of CNPq, by the RFBR project 18-02-00149-a and by the Tomsk State University Competitiveness Improvement Program.


\begin{thebibliography}{99}
\bibitem{Adler:1969gk} 
  S.~L.~Adler,
  ``Axial vector vertex in spinor electrodynamics,''
  Phys.\ Rev.\  {\bf 177}, 2426 (1969).
  doi:10.1103/PhysRev.177.2426

\bibitem{AlvarezGaume:1984nf} 
  L.~Alvarez-Gaume, S.~Della Pietra and G.~W.~Moore,
  ``Anomalies and Odd Dimensions,''
  Annals Phys.\  {\bf 163}, 288 (1985).
  doi:10.1016/0003-4916(85)90383-5

\bibitem{Atiyah:1975jf}
  M.~F.~Atiyah, V.~K.~Patodi and I.~M.~Singer,
  ``Spectral asymmetry and Riemannian Geometry 1,''
  Math.\ Proc.\ Cambridge Phil.\ Soc.\  {\bf 77}, 43 (1975).
  doi:10.1017/S0305004100049410

\bibitem{BV}
A.~O.~Barvinsky and G.~A.~Vilkovisky,
  ``The Generalized Schwinger-Dewitt Technique in Gauge Theories and Quantum Gravity,''
  Phys.\ Rept.\  {\bf 119}, 1 (1985).
  doi:10.1016/0370-1573(85)90148-6
 
\bibitem{Bell:1969ts} 
  J.~S.~Bell and R.~Jackiw,
  ``A PCAC puzzle: $\pi^0 \to \gamma \gamma$ in the $\sigma$ model,''
  Nuovo Cim.\ A {\bf 60}, 47 (1969).
  doi:10.1007/BF02823296
  
\bibitem{BGV2004}
N.~Berline, E.~Getzler and M.~Vergne,
\emph{Heat Kernels and Dirac Operators.}
Springer, Berlin, 2004.

\bibitem{BB2013}
D.~D.~Bleecker and B.~Boo{\ss}-Bavnbek,
\emph{Index Theory with Applications to Mathematics and Physics.}
International Press, Boston, 2013.

\bibitem{Bordag:1999ed} 
  M.~Bordag and D.~V.~Vassilevich,
  ``Heat kernel expansion for semitransparent boundaries,''
  J.\ Phys.\ A {\bf 32}, 8247 (1999)
  doi:10.1088/0305-4470/32/47/304
  [hep-th/9907076].

\bibitem{DeWitt1965}
B.~S.~DeWitt, \textit{Dynamical Theory of Groups and Fields}, Gordon and Breach, NewYork, 1965.

\bibitem{Fukaya:2019qlf} 
  H.~Fukaya, M.~Furuta, S.~Matsuo, T.~Onogi, S.~Yamaguchi and M.~Yamashita,
  ``The Atiyah-Patodi-Singer index and domain-wall fermion Dirac operators,''
  arXiv:1910.01987 [math.DG].

\bibitem{Fursaev:2011zz} 
  D.~Fursaev and D.~Vassilevich,
  \emph{Operators, Geometry and Quanta : Methods of spectral geometry in quantum field theory}, Spinger, Dordrecht, 2011.
  doi:10.1007/978-94-007-0205-9
	
\bibitem{Gilkey:1984}
P.~B.~Gilkey,
\emph{Invariance theory, the heat equation, and the Atiyah-Singer index theorem.}
Publish or Perish, Wilmington, 1984

\bibitem{Gilkey:2001mj} 
  P.~B.~Gilkey, K.~Kirsten and D.~V.~Vassilevich,
  ``Heat trace asymptotics with transmittal boundary conditions and quantum brane world scenario,''
  Nucl.\ Phys.\ B {\bf 601}, 125 (2001)
  doi:10.1016/S0550-3213(01)00083-9
  [hep-th/0101105].

\bibitem{Hortacsu:1980kv} 
  M.~Hortacsu, K.~D.~Rothe and B.~Schroer,
  ``Zero Energy Eigenstates for the Dirac Boundary Problem,''
  Nucl.\ Phys.\ B {\bf 171}, 530 (1980).
  doi:10.1016/0550-3213(80)90384-3

\bibitem{2}
A.~V.~Ivanov,
  ``Diagram Technique for the Heat Kernel of the Covariant Laplace Operator,''
  Theor.\ Math.\ Phys.\  {\bf 198}, no. 1, 100-117 (2019)
  doi:10.1134/S0040577919010070
  [arXiv:1905.05455 [hep-th]].

\bibitem{Kirsten}
K.~Kirsten,
\emph{Spectral functions in mathematics and physics}, Chapman \& Hall/CRC, Boca Raton, 2001.

\bibitem{Kurkov:2017cdz} 
  M.~Kurkov and D.~Vassilevich,
  ``Parity anomaly in four dimensions,''
  Phys.\ Rev.\ D {\bf 96}, no. 2, 025011 (2017)
  doi:10.1103/PhysRevD.96.025011
  [arXiv:1704.06736 [hep-th]].
  
\bibitem{Kurkov:2018pjw} 
  M.~Kurkov and D.~Vassilevich,
  ``Gravitational parity anomaly with and without boundaries,''
  JHEP {\bf 1803}, 072 (2018)
  doi:10.1007/JHEP03(2018)072
  [arXiv:1801.02049 [hep-th]].
  
\bibitem{Moss:2000gv} 
  I.~G.~Moss,
  ``Heat kernel expansions for distributional backgrounds,''
  Phys.\ Lett.\ B {\bf 491}, 203 (2000)
  doi:10.1016/S0370-2693(00)00966-7
  [hep-th/0007185].

\bibitem{Nakahara:2003nw} 
  M.~Nakahara,
  \emph{Geometry, topology and physics},
  IoP, Bristol (2003).

\bibitem{Niemi:1983rq} 
  A.~J.~Niemi and G.~W.~Semenoff,
  ``Axial Anomaly Induced Fermion Fractionization and Effective Gauge Theory Actions in Odd Dimensional Space-Times,''
  Phys.\ Rev.\ Lett.\  {\bf 51}, 2077 (1983).
  doi:10.1103/PhysRevLett.51.2077

\bibitem{Redlich:1983dv} 
  A.~N.~Redlich,
  ``Parity Violation and Gauge Noninvariance of the Effective Gauge Field Action in Three-Dimensions,''
  Phys.\ Rev.\ D {\bf 29}, 2366 (1984).
  doi:10.1103/PhysRevD.29.2366

\bibitem{Vassilevich:2018aqu}
  D.~Vassilevich,
  ``Index Theorems and Domain Walls,''
  JHEP {\bf 1807}, 108 (2018)
  doi:10.1007/JHEP07(2018)108
  [arXiv:1805.09974 [hep-th]].

\bibitem{Witten:2015aba} 
E.~Witten,
  ``Fermion Path Integrals And Topological Phases,''
  Rev.\ Mod.\ Phys.\  {\bf 88}, no. 3, 035001 (2016)
  doi:10.1103/RevModPhys.88.035001, 10.1103/RevModPhys.88.35001
  [arXiv:1508.04715 [cond-mat.mes-hall]].

\bibitem{Witten:2019bou} 
  E.~Witten and K.~Yonekura,
  ``Anomaly Inflow and the $\eta$-Invariant,''
  arXiv:1909.08775 [hep-th].
\end{thebibliography}
\end{document}